\documentclass[journal]{IEEEtran}
\usepackage{amsmath,amssymb,amsthm}  
\usepackage[sans]{dsfont}
\usepackage{graphicx}
\usepackage{aecompl}
\usepackage[latin1]{inputenc}
\usepackage{topcapt}
\usepackage{enumerate}

\newcommand{\zerov}{{\bf 0}}

\newcommand{\Am}{{\bf A}}
\newcommand{\Bm}{{\bf B}}
\newcommand{\Cm}{{\bf C}}
\newcommand{\Dm}{{\bf D}}

\newcommand{\Hm}{{\bf H}}
\newcommand{\Id}{{\bf I}}

\newcommand{\Mm}{{\bf M}}

\newcommand{\Rm}{{\bf R}}

\newcommand{\Tm}{{\bf T}}

\newcommand{\Wm}{{\bf W}}

\newcommand{\Xm}{{\bf X}}
\newcommand{\Ym}{{\bf Y}}

\newcommand{\Ac}{{\cal A}}
\newcommand{\Bc}{{\cal B}}
\newcommand{\Cc}{{\cal C}}

\newcommand{\Gc}{{\cal G}}

\newcommand{\Mc}{{\cal M}}
\newcommand{\Nc}{{\cal N}}
\newcommand{\Oc}{{\cal O}}

\newcommand{\Rc}{{\cal R}}

\allowdisplaybreaks
\newtheoremstyle{mystyle}
{}
{}
{\itshape}
{}
{\bf}
{.}
{.5em}
{}

\theoremstyle{mystyle}
\newtheorem{mydef}{Definition}
\newtheorem{myprop}{Proposition}

\begin{document}
\title{A Novel Construction of Multi-group Decodable Space-Time Block Codes}
\author{\IEEEauthorblockN{Amr~Ismail,~\IEEEmembership{IEEE Student Member}, Jocelyn~Fiorina,~\IEEEmembership{IEEE Member} and~Hikmet~Sari,~\IEEEmembership{IEEE Fellow}}\\
\IEEEauthorblockA{Telecommunications Department, SUPELEC, F-91192 Gif-sur-Yvette, France\\
Email:$\lbrace$amr.ismail, jocelyn.fiorina, and hikmet.sari$\rbrace$@supelec.fr\\}}

\maketitle
\begin{abstract}
Complex Orthogonal Design (COD) codes are known to have the lowest detection complexity among Space-Time Block Codes (STBCs). However, the rate of square COD codes decreases exponentially with the number of transmit antennas. The Quasi-Orthogonal Design (QOD) codes emerged to provide a compromise between rate and complexity as they offer higher rates compared to COD codes at the expense of an increase of decoding complexity through partially relaxing the orthogonality conditions. The QOD codes were then generalized with the so called $g$-symbol and $g$-group decodable STBCs where the number of orthogonal groups of symbols is no longer restricted to two as in the QOD case. However, the adopted approach for the construction of such codes is based on sufficient but not necessary conditions which may limit the achievable rates for any number of orthogonal groups. In this paper, we limit ourselves to the case of  Unitary Weight (UW)-$g$-group decodable STBCs for $2^a$ transmit antennas where the weight matrices are required to be single thread matrices with non-zero entries $\in\lbrace\pm 1,\pm j\rbrace$ and address the problem of finding the highest achievable rate for any number of orthogonal groups. This special type of weight matrices guarantees full symbol-wise diversity and subsumes a wide range of existing codes in the literature. We show that in this case an exhaustive search can be applied to find the maximum achievable rates for UW-$g$-group decodable STBCs with $g>1$. For this purpose, we extend our previously proposed approach for constructing UW-2-group decodable STBCs based on necessary and sufficient conditions to the case of UW-$g$-group decodable STBCs in a recursive manner. 
\end{abstract}


\section{Introduction}
\IEEEPARstart{S}{pace}-Time Block Codes (STBCs) were originally proposed as a low-complexity alternative to Space-Time Trellis Codes (STTCs) \cite{STTC} which suffered from a prohibitively high decoding complexity. By the decoding complexity we mean the minimum number of times an exhaustive search decoder has to compute the ML metric to optimally estimate the transmitted symbols codeword \cite{FAST}. The STBCs are characterized by their linearity over the field of real numbers as the transmitted code matrix can be expressed as a linear weighted combination of real information symbols. Moreover, the STBCs may efficiently exploit the MIMO channel degrees of freedom and diversity. 

The first proposed STBC was unquestionably the elegant two transmit antennas diversity scheme known afterwards as \textit{the Alamouti code} \cite{ALAMOUTI}. In fact, the Alamouti code enables separate decoding of each complex symbol giving rise to a decoding complexity that only grows linearly with the underlying constellation size for general constellations, and it can be effectively decoded at a constant complexity irrespectively of the underlying constellation size through hard PAM slicers in the case of rectangular QAM constellations such as 4/16-QAM. 

In an attempt to generalize Alamouti's low-complexity scheme for a number of transmit antennas greater than two, the well-known family of low-complexity decodable codes namely orthogonal STBCs has been proposed \cite{HR_ORTH,CLIFF_ORTH}. Unfortunately, the rate of square orthogonal STBCs decays exponentially with the number of transmit antennas \cite{CLIFF_ORTH}, which makes them more suitable to low-rate communications. Arguably, the first proposed low-complexity rate-1 code for the case of four transmit antennas is the Quasi-Orthogonal (QO)STBC originally proposed by H. Jafarkhani \cite{QOSTBC} and later optimized through constellation rotation to provide full diversity \cite{FULL_DIVERSITY_QOD,OPT_QORTH_ROT}. The QOSTBC partially relaxes the orthogonality conditions by allowing two complex symbols to be jointly detected. Subsequently, rate-one, full-diversity QOSTBCs were proposed for an arbitrary number of transmit antennas that subsume the original QOSTBC as a special case \cite{GEN_QOSTBC}. In this general framework, the \textit{quasi-orthgonality} stands for decoupling the transmitted symbols into two groups of the same size. 

However, STBCs with lower decoding complexity may be obtained through the concept of multi-group decodability laid by the S. Karmakar \textit{et al.} in \cite{MULTI_SYMs,MULT_GR}. Indeed, the multi-group decodability generalizes the quasi-orthogonality by allowing the codeword of symbols to be decoupled into more than two groups not necessarily of the same size. Moreover, one can obtain rate-one, full-diversity 4-group decodable STBCs for an arbitrary number of transmit antennas \cite{SAST}.

On the other hand, the adopted approach for the construction of the $g$-symbol (resp. $g$-group)  decodable codes (namely the Clifford Unitary Weight (CUW)-$g$-symbol \cite{MULTI_SYMs} (resp. $g$-group \cite{MULT_GR}) decodable codes) is based on sufficient but not necessary conditions which may limit the achievable rate for any number of orthogonal groups. In \cite{COLL&DIST} the authors have found the maximal achievable rate for the CUW-$\lambda$-symbol decodable codes but only for the case where the number of symbols $\lambda$ in each group can be expressed as $2^a,\ a\in \mathbb{N}$. In \cite{2SYMS_BND} the authors proved that the rate of arbitrary $2^a \times 2^a$ UW-Single-symbol decodable STBCs is upper bounded by $\frac{a}{2^{a-1}}$. Consequently, the question on the maximum achievable rate for an arbitrary number of orthogonal groups remains open. 
 
In this paper, we limit ourselves to a special case of UW-$g$-group decodable STBCs for $2^a$ transmit antennas where the weight matrices are required to be single thread matrices with non-zero entries $\in \lbrace\pm 1,\pm j\rbrace$ and address the problem of finding the highest achievable rate for any number of orthogonal groups. This special type of weight matrices guarantees full symbol-wise diversity \cite{SWDIV} and subsumes a wide range of existing codes in the literature (a non-exhaustive list includes \cite{GEN_QOSTBC}-\cite{2G_2m}). For this purpose we extend the approach proposed in \cite{GLOB10} for constructing UW-2-group decodable STBCs based on necessary and sufficient conditions to the case of UW-$g$-group decodable STBCs in a recursive manner. 

The major idea is that contrary to what was done in \cite{NUM_QORTH}, we are dealing with the $\pmb{\Lambda}_{kl}$ matrix where $\pmb{\Lambda}_{kl}=\Am_k\Bm_l^H$ instead of dealing directly with the weight matrices $\Am_k$ and $\Bm_l$. This approach reduced the number of candidate  $\pmb{\Lambda}_{kl}$ matrices that can be used for the construction of the UW-$g$-group decodable STBCs, as they have to satisfy additional properties over those of the weight matrices. Then, the vector space representation is used to build the $\pmb{\Lambda}_{kl}$ matrices and we show that the number of candidate $\pmb{\Lambda}_{kl}$ matrices becomes limited (see Appendix). A search routine can then be applied to find existing UW-$g$-group decodable codes at given rates which enables us to determine the maximum achievable rate for an arbitrary number of orthogonal groups. UW-$g$-group decodable codes for a number of transmit antennas that is not a power of two can be easily obtained by the removal of an appropriate number of columns of the code matrix corresponding to the nearest greater number of transmit antennas that is a power of two.

The paper is organized as follows: In the next section, preliminaries about $g$-group decodable codes are provided. In Section III, we reformulate the $g$-group decodability conditions in terms of the $\pmb{\Lambda}$ matrices. Section IV addresses the construction of UW-$g$-group decodable codes. In Section V, we present the results of the exhaustive search of UW-$g$-group decodable codes for four transmit antennas based on the construction method developed in the former section, and finally we give our conclusions in Section VI. 
  
\subsection*{Notations}
In this paper, small letters denote scalar variables, bold small letters are used to designate vectors, and bold capital letters are used to designate matrices. If $\Am$ is a matrix, then $\Am^H$ and $\Am^T$ denote the hermitian and the transpose of $\Am$, respectively. $\Id$ and $\zerov$ denote the identity and the null matrices, respectively. The $\text{tr}\left(\Am\right)$, $\text{det}\left(\Am\right)$ and $\Vert\Am\Vert_F$ denote the trace, the determinant and the Frobenius norm respectively of $\Am$. $\left(m\right)_k$ means $m$ modulo $k$, $\Mc_n$ is the set of $n\times n$ complex matrices and for a finite set $\Ac$, $\vert\mathcal{A}\vert$ denotes its cardinality.

\section{Preliminaries}
We define the MIMO channel input-output relationship as: 
\begin{equation}
\underset{T\times N_r}{\Ym} =\underset{T\times N_t}{\Xm} \underset{N_t\times N_r}{\Hm} +\underset{T\times N_r}{\Wm}
\label{model}  
\end{equation} 
where $T$ is the codeword signalling period, $N_r$ is the number of receive antennas, $N_t$ is the number of transmit antennas, $\Ym$ is the received signal matrix, $\Xm$ is the code matrix, $\Hm$ is the channel coefficients matrix with entries $h_{kl} \sim \Cc \Nc(0,1)$, and $\Wm$ is the noise matrix with entries $w_{ij} \sim \Cc \Nc(0,N_{0})$. According to the above model, the $t$'th row of $\Xm$ denotes the symbols transmitted through the $N_t$ transmit antennas during the $t$'th channel use while the $n$'th column denotes the symbols transmitted through the $n$'th transmit antenna during the codeword signalling period $T$.

A STBC matrix $\Xm$ that encodes $2K$ real symbols can be expressed as a linear combination of the transmitted symbols as \cite{LDSTBC}:
\begin{equation}
\Xm=\sum^{2K}_{k=1} \Am_{k}x_k
\label{LD}
\end{equation}
with the symbols $x_k \in \mathbb{R}$, and the $\Am_k, k=1,...,2K$ are $T \times N_t$ complex matrices called dispersion or weight matrices that are required to be linearly independent over $\mathbb{R}$.

Multi-group decodable STBCs are designed to significantly reduce the decoding complexity by allowing separate detection of disjoint groups of symbols without any loss of performance. This is achieved iff the ML metric can be expressed as a sum of terms involving independent groups of symbols. This suggests the following definition: 
\begin{mydef}A STBC that encodes $2K$ real symbols is said to be $g$-group decodable iff its ML metric can be expressed as a sum of $g$ terms that depend on disjoint subsets of the transmitted symbols \cite{MULTI_SYMs,MULT_GR}.
\end{mydef}
The conditions to be satisfied by the weight matrices of a STBC in order to be $g$-group decodable are  derived in \cite{MULTI_SYMs,MULT_GR} from the ML decision rule of the system model \eqref{model} and are outlined here for self-completeness. Assuming that perfect CSI is available at the receiver side, the ML estimated codeword is given by:
\begin{equation}
\begin{split}
\Xm^{\text{ML}}&=\text{arg}\ \underset{\Xm \in \Cc}{\text{min}}\ \Vert \Ym-\Xm \Hm \Vert^2_F\\
&=\text{arg}\ \underset{\Xm \in \Cc}{\text{min}}\ \text{tr}\left(\left(\Ym-\Xm \Hm \right)^H\left(\Ym-\Xm \Hm \right)\right)  
\end{split}
\label{frobenius} 
\end{equation}
where $\Cc$ denotes the codebook. If $\Xm$ can be expressed as a sum of sub-codes $\Xm_i,\ i=1,\ldots g$ such that:
\begin{equation}
\Xm=\sum^{g}_{i=1} \Xm_i,\ \Xm^H_i\Xm_j+\Xm^H_j\Xm_i=\zerov ,\ 1\leq i\neq j \leq g
\label{subcodes}
\end{equation}
\eqref{frobenius} reduces to:
\begin{equation*}
\begin{split}
&\text{arg}\ \underset{\Xm \in \Cc}{\text{min}}\ \text{tr}\left(\Ym^H\Ym -\sum^{g}_{i=1} \Ym^H\Xm_i\Hm-\Hm^H\Xm^H_i\Ym+\Hm^H\Xm^H_i\Xm_i\Hm\right)\\
&=\sum^{g}_{i=1}\text{arg}\ \underset{\Xm_i \in \Cc_i}{\text{min}}\ \text{tr}\left(\Ym^H\Ym - \Ym^H\Xm_i\Hm-\Hm^H\Xm^H_i\Ym+\Hm^H\Xm^H_i\Xm_i\Hm\right)\\
&-\sum^{g-1}_{i=1}\text{tr}\left(\Ym^H\Ym\right)
\end{split}
\end{equation*}
where $\Cc_i$ denotes the codebook of the $i$'th sub-code. Noting that the last term of the above is constant for a given received signal matrix, the ML decision rule may be expressed as:
\begin{equation}
\Xm^{\text{ML}}=\sum^{g}_{i=1}\text{arg}\ \underset{\Xm_i \in \Cc_i}{\text{min}}\ \Vert \Ym-\Xm_i\Hm \Vert^2_F.
\end{equation}
In terms of weight matrices, it is straightforward to verify that \eqref{subcodes} is equivalent to:
\begin{equation}
\begin{split}
\Am_k^{H}\Am_l+\Am_l^{H}\Am_k=\pmb{0},\ \forall \Am_k \in \Gc_i,\ \Am_l \in \Gc_j,\\ 1\leq i\neq j\leq g,\ \vert \Gc_i \vert=n_i,\sum^{g}_{i=1}n_i=2K.
\end{split}
\label{g-group}
\end{equation}
where $\Gc_i$ is the set of weight matrices associated to the $i$'th group of symbols. If a STBC that encodes $2K$ real symbols is $g$-group decodable, its decoding complexity order can be reduced from $M^K$ to $\sum^{g}_{i=1}\sqrt{M}^{n_i}$ where $M$ is the size of the used square QAM constellation. If a real Sphere Decoder (SD) , $g$-group decodability reduces to splitting the original tree with $2K$ levels to $g$ smaller trees each with $n_i$ levels. The decoding complexity order can be further reduced to $\sum^{g}_{i=1}\sqrt{M}^{n_i-1}$ if the conditional detection with hard slicer is employed. In the special case of orthogonal STBCs, the decoding complexity is $\Oc(1)$, as the PAM slicers need only a fixed number of arithmetic operations irrespectively of the square QAM constellation size. 
 
Besides its induced significant reduction in the worst-case decoding complexity, the multi-group decodability structure enables a simplified coding gain optimization as the global coding gain optimization problem turns into the optimization of the individual coding gain of each sub-code, which is illustrated in the following proposition:
\begin{myprop} If a STBC $\Xm$ is $g$-group decodable as:
\begin{equation*}
\Xm=\sum^{g}_{i=1} \Xm_i,\ \Xm^H_i\Xm_j+\Xm^H_j\Xm_i=\zerov ,\ 1\leq i\neq j \leq g
\end{equation*}
its coding gain $\delta_{\Xm}$ is expressed as:
\begin{equation*}
\delta_{\Xm}=\rm{min} \left\lbrace \delta_{\Xm_1},\delta_{\Xm_2},\ldots,\delta_{\Xm_g}\right\rbrace 
\end{equation*} 
where $\delta_{\Xm_i}$ denotes the coding gain of $i$'th sub-code.
\end{myprop}
\begin{proof}
The proof follows directly from \eqref{subcodes} and the Minkowski determinant inequality \cite{MATRIX_ANALYS}. Recalling the coding gain definition, one has:
\begin{equation*}
\begin{split}
\delta_{\Xm}&=\underset{\underset{\Xm,\Xm'\in\Cc}{\Xm\neq\Xm'}}{\text{min}}\ \text{det} \left(\left(\Xm-\Xm'\right)^H\left(\Xm-\Xm'\right)\right)\\
&=\underset{\Delta \Xm\in\Delta\Cc/\left\lbrace \zerov\right\rbrace}{\text{min}}\ \text{det} \left(\left(\Delta \Xm\right)^H\left(\Delta \Xm\right)\right)  
\end{split}
\end{equation*}  
Thanks to \eqref{subcodes}, the above reduces to:
\begin{equation}
\delta_{\Xm}=\underset{\Delta \Xm\in\Delta\Cc/\left\lbrace \zerov\right\rbrace}{\text{min}}\ \text{det}\left(\sum^{g}_{i=1} \Delta\Xm^H_i\Delta\Xm_i\right)
\end{equation}
but from Minkowski's determinant inequality \cite{MATRIX_ANALYS} one can write:
\begin{equation}
\left(\text{det}\left(\Am+\Bm\right)\right)^{1/n}\geq \left( \text{det}\left(\Am\right)\right)^{1/n} +\left( \text{det}\left(\Bm\right)\right)^{1/n}  
\label{minkowski}
\end{equation}
where $\Am,\Bm \in \Mc_n$ are positive definite matrices. Therefore:
\begin{equation*}
\begin{split}
\text{det}\left(\sum^{g}_{i=1} \Delta\Xm^H_i\Delta\Xm_i\right)\geq &\left(\sum^{g}_{i=1} \left(\text{det}\left(\Delta\Xm^H_i\Delta\Xm_i\right)\right)^{1/n}\right)^n\\
=& \sum^{g}_{i=1} \text{det}\left(\Delta\Xm^H_i\Delta\Xm_i\right)+C\\
\geq & \sum^{g}_{i=1} \text{det}\left(\Delta\Xm^H_i\Delta\Xm_i\right) 
\end{split}
\end{equation*}
where the last inequality follows from the fact that $C \geq 0$. Equality holds for the trivial case 
$\Delta \Xm=\zerov$ or $\Delta \Xm=\Delta \Xm_k$ and $\Delta \Xm_i=\zerov\ \forall 1\leq i \neq k \leq g$. Thus we have:
\begin{equation}
\begin{split}
\delta_{\Xm}=&\text{min}\Big\lbrace \underset{\Delta \Xm_1\in\Delta\Cc_1/\left\lbrace \zerov\right\rbrace}{\text{min}}\text{det}\left(\Delta\Xm^H_1\Delta\Xm_1\right),\\
&\ldots,\underset{\Delta \Xm_g\in\Delta\Cc_g/\left\lbrace \zerov\right\rbrace}{\text{min}}\text{det}\left(\Delta\Xm^H_g\Delta\Xm_g\right)\Big\rbrace\\ 
=&\text{min} \left\lbrace \delta_{\Xm_1},\delta_{\Xm_2},\ldots,\delta_{\Xm_g}\right\rbrace 
\end{split}
\end{equation}
which concludes the proof.
\end{proof}
If the weight matrices are all unitary, the code is called Unitary Weight (UW)-$g$-group decodable, and if $n_1=n_2=\ldots=n_g$, the $g$-group decodable code is said to be symmetric.   

\section{Necessary and sufficient conditions for $g$-group decodability}
\subsection{UW-2-group decodability}
In order to construct a UW-2-group decodable code which transmits $2K$ real symbols, we must find $2$ sets of unitary weight matrices namely $(\Gc_1,\Gc_2)$ that are linearly independent over $\mathbb{R}$ such that $n_1+n_2=2K$, and each pair of weight matrices belonging to different sets must satisfy \eqref{g-group}. In the following, we will reformulate the problem in a way which allows exhaustive search of weight matrices for 2-groups decodable codes. We will call the $k$'th weight matrix of the first group $\Am_k$ and the $l$'th weight matrix of the second group $\Bm_l$.

Multiplying \eqref{g-group} from the left by $\Am_k$ and from the right by $\Bm^H_l$:
\begin{equation}
\Am_k\left(\Am^H_k\Bm_l+\Bm^H_l\Am_k\right)\Bm^H_l=\zerov 
\end{equation} 
which can be written as:
\begin{equation}
(\underbrace{\Am_k\Bm^H_l}_{\pmb{\Lambda}_{kl}})^2=-\Id_T
 \label{lambdakl_1}
\end{equation}
where  $\Id_T$ is the $T\times T$ identity matrix. Thus $\pmb{\Lambda}_{kl}$ must be a unitary matrix squaring to $-\Id_T$.

$\pmb{\Lambda}_{kl}$ can be expressed as:
\begin{eqnarray}
\pmb{\Lambda}_{kl}=\Am_k\Bm^H_l\nonumber&=&\Am_k\Bm^H_{1}\Bm_{1}\Bm^H_l\nonumber\\
&=&\pmb{\Lambda}_{k1}\Bm_{1}\Am^H_{1}\Am_{1}\Bm^H_l\nonumber\\
&=&\pmb{\Lambda}_{k1}\pmb{\Lambda}^H_{11}\pmb{\Lambda}_{1l}\nonumber\\
&=&-\pmb{\Lambda}_{k1}\pmb{\Lambda}_{11}\pmb{\Lambda}_{1l}
\label{lambdakl_2}
\end{eqnarray}
In the last step, we used the fact that a unitary matrix that squares to $-\Id$ is anti-hermitian.
\begin{myprop}
A code is said to be UW-2-group decodable code iff:
\begin{enumerate}
\item $\Gamma=\left\lbrace \pmb{\Lambda}_{11},\pmb{\Lambda}_{k1},\pmb{\Lambda}_{1l}:2 \leq k\leq n_1,2 \leq l\leq n_2\right\rbrace$ is a set of unitary matrices that square to $-\Id$;                                                     
\item $\left(\pmb{\Lambda}_{k1}\pmb{\Lambda}_{11}\pmb{\Lambda}_{1l}\right)^2=-\Id\ \forall\ 2\leq k\leq n_1,\ 2\leq l\leq n_2$;
\item the set $\left\lbrace \pmb{\Lambda}_{k1},\pmb{\Lambda}_{1l}\pmb{\Lambda}_{11}:1\leq k\leq n_1,1\leq l\leq n_2\right\rbrace$ is linearly independent over $\mathbb{R}$.
\end{enumerate}
where $n_i=\vert\Gc_i\vert$.
\end{myprop}
\begin{proof}
The first and the second conditions are necessary and sufficient in order to satisfy \eqref{g-group} and follow directly from \eqref{lambdakl_1} and \eqref{lambdakl_2}, respectively. The last condition is necessary and sufficient to guarantee that the weight matrices are linearly independent over $\mathbb{R}$. We will prove the last condition by proving that the linear dependence of the weight matrices over $\mathbb{R}$ implies that the set $\left\lbrace \pmb{\Lambda}_{k1},\pmb{\Lambda}_{1l}\pmb{\Lambda}_{11}:1\leq k\leq n_1,1\leq l\leq n_2\right\rbrace$ is linearly dependent over $\mathbb{R}$ and vice versa. To this end let us suppose that:
\begin{equation}
\sum^{n_1}_{k=1}a_k\Am_k+\sum^{n_2}_{l=1}b_l\Bm_l=\zerov
\label{LD}
\end{equation}
where $\lbrace a_k,b_l:1\leq k \leq n_1,1\leq l \leq n_2\rbrace \in \mathbb{R}$. Right multiplying the above by $\Bm_1^H$, we obtain:
\begin{eqnarray}
\eqref{LD}&\Leftrightarrow&\left(\sum^{n_1}_{k=1}a_k\Am_k+\sum^{n_2}_{l=1}b_l\Bm_l\right)\Bm^H_1=\zerov\\
&=&\sum^{n_1}_{k=1}a_k\pmb{\Lambda}_{k1}+\sum^{n_2}_{l=1}b_l\Bm_l\Bm^H_1=\zerov\\
&\Leftrightarrow&\sum^{n_1}_{k=1}a_k\pmb{\Lambda}_{k1}+\sum^{n_2}_{l=1}b_l\pmb{\Lambda}_{1l}^H\Am_1\Am^H_1\pmb{\Lambda}_{11}=\zerov\\
&=&\sum^{n_1}_{k=1}a_k\pmb{\Lambda}_{k1}-\sum^{n_2}_{l=1}b_l\pmb{\Lambda}_{1l}\pmb{\Lambda}_{11}=\zerov
\end{eqnarray}
which means that the weight matrices are linearly independent over $\mathbb{R}$ iff the set\\ $\left\lbrace \pmb{\Lambda}_{k1},\pmb{\Lambda}_{1l}\pmb{\Lambda}_{11}:1\leq k\leq n_1,1\leq l\leq n_2\right\rbrace$ is linearly independent over $\mathbb{R}$.\end{proof}
In order to construct UW-2-group decodable codes, we will search for matrices satisfying \textbf{Proposition 2}. Once we have a set $\Gamma$ that satisfies \textbf{Proposition 2}, the corresponding UW-2-group decodable code is built as follows. First, we choose an arbitrary unitary matrix $\Am_1$ then the STBC weight matrices are obtained according to:
\begin{equation}
\begin{split}
\Bm_l &= \pmb{\Lambda}^H_{1l}\Am_1;\ 1\leq l\leq n_2\\
\Am_k &= \pmb{\Lambda}_{k1}\Bm_1;\ 2\leq k\leq n_1
\end{split}
\end{equation} 
which means that for a given set $\Gamma$ the corresponding UW-2-group decodable code is not unique.
\subsection{UW-$g$-group decodability}
In order to expand the above approach to UW-$g$-group decodable codes, it is worth noting that any $g$-group decodable code can be seen as a 2-group decodable code. This means that we can search for $g$-group decodable codes by iteratively searching for 2-group decodable codes. For instance if we search for a 3-group decodable code, where $n_1,n_2$ and $n_3$ are the number of weight matrices (or alternatively real symbols) in the first, second and third group respectively, then we can proceed in two steps as follows:
\begin{enumerate}[I-]
\item We search for a 2-group decodable code with $n_1$ and $(n_2+n_3)$ real symbols in the first and second group, respectively.
\item Among the found second group with $(n_2+n_3)$  weight matrices, we will search for $2$-group 
decodable codes with $n_2$ and $n_3$ weight matrices in the first and second group, respectively.
\end{enumerate}
Mathematically speaking, in the first step we search for all the sets: \[\Gamma=\left\lbrace \pmb{\Lambda}_{11},\pmb{\Lambda}_{k1},\pmb{\Lambda}_{1l}:2 \leq k\leq n_1,2 \leq l\leq (n_2+n_3)\right\rbrace\] that satisfy \textbf{Proposition 2}. In the second step, we will search among the sets $\left\lbrace \pmb{\Lambda}_{1l}:1 \leq l\leq (n_2+n_3)\right\rbrace$ for sets that can be divided into two groups of $n_2$ and $n_3$ matrices.
\begin{myprop}
A code is said to be UW-$g$-group decodable code iff:
\begin{enumerate}
\item $\Gamma=\left\lbrace\pmb{\Lambda}_{11},\pmb{\Lambda}_{k1},\pmb{\Lambda}_{1l}:2\leq k\leq n_1,2\leq l\leq \sum^{g}_{i=2}n_i\right\rbrace$ is a set of unitary matrices squaring to $-\Id$;
\item $\left(\pmb{\Lambda}_{k1}\pmb{\Lambda}_{11}\pmb{\Lambda}_{1l}\right)^2=-\Id\ \forall\ 2\leq k\leq n_1,\ 2\leq l\leq \sum^{g}_{i=2}n_i$;
\item The set $\left\lbrace \pmb{\Lambda}_{k1},\pmb{\Lambda}_{1l}\pmb{\Lambda}_{11}:1\leq k\leq n_1, 1\leq l\leq \sum^{g}_{i=2}n_i\right\rbrace$ is linearly independent over $\mathbb{R}$;
\item $\pmb{\Lambda}_{1l}\pmb{\Lambda}_{1l'}=-\pmb{\Lambda}_{1l'}\pmb{\Lambda}_{1l}\ \forall \sum^{L-1}_{i=2}n_i+1\leq l \leq \sum^{L}_{i=2}n_i, \sum^{L'-1}_{j=2}n_j+1\leq l' \leq \sum^{L'}_{j=2}n_j,\ 2\leq L\neq L'\leq g$.
\end{enumerate}
where $n_i=\vert\Gc_i\vert$.
\end{myprop} 
\begin{proof}
The first three conditions are the same as in \textbf{Proposition 2}, and therefore we need only to prove the last one. Consider a UW-2-group decodable code with $\Gc_1=\lbrace\Am_1,\ldots,\Am_{n_1}\rbrace$ and $\overset{g}{\underset{m=2}{\cup}}\Gc_m=\lbrace\Bm_1,\ldots,\Bm_{\sum^{g}_{i=2}n_i}\rbrace$, such that:
\begin{equation}
\Bm_l^{H}\Bm_{l'}+\Bm^{H}_{l'}\Bm_l=\zerov,\ \forall\ \Bm_l \in \Gc_L,\ \Bm_{l'} \in \Gc_{L'},\ 2\leq L\neq L'\leq g.
\end{equation} 
By left and right multiplying the above by $\Am_1$ and $\Am^H_1$ respectively, we obtain: 
\begin{eqnarray}
\Am_1\Bm_l^H\Bm_{l'}\Am^H_1+\Am_1\Bm^H_{l'}\Bm_l\Am^H_1=\zerov\\
\pmb{\Lambda}_{1l}\pmb{\Lambda}^H_{1l'}+\pmb{\Lambda}_{1l'}\pmb{\Lambda}^H_{1l}=\zerov\\
\pmb{\Lambda}_{1l}\pmb{\Lambda}_{1l'}+\pmb{\Lambda}_{1l'}\pmb{\Lambda}_{1l}=\zerov.
\end{eqnarray}
which concludes the proof.
\end{proof} 
The weight matrices are obtained as in the UW-2-group decodable code. We first choose an arbitrary unitary matrix $\Am_1$ and then:
\begin{equation}
\begin{split}
\Bm_l &= \pmb{\Lambda}^H_{1l}\Am_1;\ 1\leq l\leq \sum^{g}_{i=2}n_i\\
\Am_k &= \pmb{\Lambda}_{k1}\Bm_1;\ 2\leq k\leq n_1
\end{split}
\end{equation}


\section{Construction of the matrices in $\Gamma$}
Recognizing the key role of the matrix representations of the generators of the Clifford Algebra over $\mathbb{R}$ in the sequel of this paper, we briefly review their main properties in the following. 
\subsection{Linear Representations of Clifford generators}
The defining relation of the generators of the Clifford algebra over $\mathbb{R}$ is:
\begin{equation}
\gamma_i\gamma_j+\gamma_j\gamma_i=-2\delta_{ij}\mathds{1}
\label{cliff_gen}
\end{equation} 
The above equation can be split in two equations:
\begin{eqnarray}
\gamma_i^2&=&-\mathds{1}\\
\gamma_i\gamma_j&=&-\gamma_j\gamma_i\ \forall i\neq j.
\end{eqnarray}
In \cite{CLIFF_ORTH}, the question about the maximum number of unitary representations of the Clifford generators has been thoroughly addressed and it has been proven that for $2^a \times 2^a$ matrices there is exactly $2a+1$ unitary matrix representations of the Clifford algebra generators. The matrix representations of $\gamma_i$ denoted $\Rc\left(\gamma_i\right)$ for the $2^a \times 2^a$ case are obtained as \cite{CLIFF_ORTH}:
\begin{eqnarray}
\Rc(\gamma_1)&=&\pm j \underbrace{\sigma_3 \otimes \sigma_3\ldots \otimes\sigma_3}_{a}\nonumber\\
\Rc(\gamma_2)&=&\Id_{2^{a-1}}\otimes \sigma_1\nonumber\\
\Rc(\gamma_3)&=&\Id_{2^{a-1}}\otimes \sigma_2\nonumber\\
&\vdots&\nonumber\\
\Rc(\gamma_{2k})&=&\Id_{2^{a-k}}\otimes \sigma_1\underbrace{\otimes \sigma_3\otimes \sigma_3\ldots\otimes \sigma_3}_{k-1}\nonumber\\
\Rc(\gamma_{2k+1})&=&\Id_{2^{a-k}}\otimes \sigma_2\underbrace{\otimes \sigma_3\otimes \sigma_3\ldots\otimes \sigma_3}_{k-1}\nonumber\\
&\vdots&\nonumber\\
\Rc(\gamma_{2a})&=&\sigma_1\otimes \underbrace{\sigma_3 \otimes \sigma_3\ldots\sigma_3}_{a-1}\nonumber\\
\Rc(\gamma_{2a+1})&=&\sigma_2\otimes \underbrace{\sigma_3 \otimes \sigma_3\ldots\sigma_3}_{a-1}
\label{reps}
\end{eqnarray}
where
\begin{equation}
\sigma_1=\begin{bmatrix}0&1\\-1&0\end{bmatrix},\ \sigma_2=\begin{bmatrix}0&j\\j&0\end{bmatrix},\ \sigma_3=\begin{bmatrix}1&0\\0&-1\end{bmatrix}
\label{sigmas}
\end{equation}
From now on, we will denote $\Rc(\gamma_i)$ by $\Rm_i$ for simplicity. The properties of the matrices $\Rm_i,\ i=1,\ldots,2a+1$ can then be summarized as follows:
\begin{equation}
\begin{split}
\Rm_i^{H}=-\Rm_i;\ \Rm_i^2=-\Id;&\ \text{and}\ \Rm_i\Rm_j+\Rm_j\Rm_i=\pmb{0},\\ \forall 1\leq i&\neq j\leq 2a+1.
\end{split}
\label{properties}
\end{equation}
However, one has from \cite{ANTICOMMUTING} that if $\left\lbrace \Mm_k:k=1,\ldots,2a\right\rbrace$ are pairwise anti-commuting matrices that square to a scalar, then the set:
\begin{equation*}
\begin{split}
\left\lbrace\Id\right\rbrace \cup&\left\lbrace \Mm_k:k=1,\ldots,2a\right\rbrace\\
\overset{2a}{\underset{m=2}{\cup}}&\left\lbrace \prod^{m}_{i=1}\Mm_{k_i}:1\leq k_1 < k_2 \ldots < k_m \leq 2a\right\rbrace
\end{split}
\end{equation*}
forms a basis of $\Mc_{2^a}$ over $\mathbb{C}$. Consequently, thanks to the properties of the matrix representations of Clifford algebra generators \eqref{properties}, the set of matrices defined in \eqref{Cbasis} forms a basis of $\Mc_{2^a}$ over $\mathbb{C}$:
\begin{equation}
\begin{split}
\left\lbrace\Id\right\rbrace \cup&\left\lbrace \Rm_k:k=1,\ldots,2a\right\rbrace\\
\overset{2a}{\underset{m=2}{\cup}}&\left\lbrace \prod^{m}_{i=1}\Rm_{k_i}:1\leq k_1 < k_2 \ldots < k_m \leq 2a\right\rbrace
\end{split}
\label{Cbasis}
\end{equation}
We require that all the basis elements are anti-hermitian in order to facilitate the search of matrices in $\Gamma$ as will be shown shortly. This may be achieved by replacing $\Id$ by $j\Id$ and multiplying the basis elements by $j^{\delta\left(m\right)}$ where $\delta\left(m\right)=\frac{\left(\left(m\right) _4-1\right)\left(\left( m\right)_4-2\right)}{2}$ which does not alter the linear independence over $\mathbb{C}$. Therefore the new basis is: 
\begin{equation}
\begin{split}
j\Id\cup&\left\lbrace \Rm_k:k=1,\ldots,2a\right\rbrace\\
\overset{2a}{\underset{m=2}{\cup}}&\left\lbrace j^{\delta(m)}\prod^{m}_{i=1}\Rm_{k_i}:1\leq k_1 < k_2 \ldots < k_m \leq 2a\right\rbrace
\end{split}
\label{AHbasis} 
\end{equation}
The matrices belonging to the above basis may be easily verified to be anti-hermitian by noting that for $\Am=j^{\delta\left(m\right)}\prod^{m}_{i=1}\Rm_{k_i}\vert k_1<k_2\ldots<k_m$, we have:
\begin{equation}
\Am^{H}=(-1)^{\delta\left(m\right)+m(m+1)/2}\Am
\label{ishermitian}
\end{equation}
and it is straightforward to verify that $\delta\left(m\right)+m(m+1)/2$ is odd irrespectively of $m$.\\
For instance, the basis elements of $\Mc_4$ over $\mathbb{C}$ denoted by $(\pmb{\alpha}_i:1 \leq i\leq 16)$ are expressed in Table~\ref{basis_mat}.
\begin{myprop}
The properties of the basis elements of $\Mc_n,n=2^a$ over $\mathbb{C}$ can be summarized as follows:
\begin{enumerate}
\item $\left(\pmb{\alpha}_i\right)^2=-\Id,\ \forall 1\leq i \leq n^2$
\item $\pmb{\alpha}^H_i=-\pmb{\alpha}_i,\ \forall 1\leq i \leq n^2$
\item $\pmb{\alpha_i}\pmb{\alpha_j}=\pm \pmb{\alpha_j}\pmb{\alpha_i},\ \forall i \neq j$
\item $\pmb{\alpha}_k\pmb{\alpha}_l=\lambda\pmb{\alpha}_m$, where $\lambda \in \left\lbrace \pm 1,\pm j\right\rbrace \forall 1\leq k \neq l \leq n^2$ and $2\leq m \leq n^2$.
\end{enumerate}
\end{myprop}
\begin{proof}
The first three properties follow directly from \eqref{properties}. The latter property can be verified easily from \eqref{properties} and \eqref{AHbasis}. It remains only to verify that the product of any pair of distinct basis elements is not proportional to the identity. Let us suppose that $\pmb{\alpha}_k\pmb{\alpha}_l=\lambda\Id$ with $k \neq l$, then:
\begin{eqnarray}
\pmb{\alpha}_k\pmb{\alpha}_l+\lambda\pmb{\alpha}_k\pmb{\alpha}_k=\zerov\nonumber\\
\pmb{\alpha}_k\left(\pmb{\alpha}_l+\lambda\pmb{\alpha}_k\right)=\zerov
\end{eqnarray}
As all the basis elements are unitary matrices (thus of full rank), the only solution to the above equation is $\pmb{\alpha}_l+\lambda\pmb{\alpha}_k=\zerov$, which contradicts the linear independence property.
\end{proof}

\subsection{Necessary conditions for the matrices in $\Gamma$}
\begin{myprop}
Let the matrix $\pmb{\Lambda}$ be written as a linear combination of the basis elements in \eqref{AHbasis} as below:
\begin{equation}
\pmb{\Lambda}=\sum^{n^2}_{i=1} a_i\pmb{\alpha}_i
\label{rep1}
\end{equation}
Then, $\pmb{\Lambda}$ is unitary and squares to $-\Id$ iff  for $i \in \left\lbrace 1,2,\ldots,n^2\right\rbrace,\ a_i \in \mathbb{R}$  with $\sum^{n^2}_{i=1} a_i^2=1$ and the sum over the product of commuting pairs of basis elements equals to $\zerov$.
\end{myprop} 
\begin{proof}
$\pmb{\Lambda}$ is required to be anti-hermitian and to square to $-\Id$:
\begin{eqnarray}
\sum^{n^2}_{i=1} a_i\pmb{\alpha}_i+\sum^{n^2}_{i=1} a^*_i\pmb{\alpha}^H_i&=&\nonumber
\sum^{n^2}_{i=1} a_i\pmb{\alpha}_i-\sum^{n^2}_{i=1} a^*_i\pmb{\alpha}_i\\
&=&\sum^{n^2}_{i=1} \left(a_i-a^*_i\right)\pmb{\alpha}_i=\zerov
\end{eqnarray}
From the linear independence property of the basis, the only solution to the above equation is that $a_i=a^*_i\ \forall 1\leq i \leq n^2 $ which proves the first claim of our proposition.
\begin{eqnarray}
\left(\sum^{n^2}_{k=1} a_k\pmb{\alpha}_k\right)^2+\Id&=&  \nonumber
\left(\sum^{n^2}_{k=1}a_k\pmb{\alpha}_k\right)\left(\sum^{n^2}_{l=1}a_l\pmb{\alpha}_l\right)+\Id\\ \nonumber
&=&\sum^{n^2}_{k=1}a^2_k\left(\pmb{\alpha}_k\right)^2+\underset{l\neq k}{\sum^{n^2}_{k=1}}
a_ka_l\pmb{\alpha}_k\pmb{\alpha}_l+\Id\\ \nonumber
&=&\left(1-\sum^{n^2}_{k=1}a^2_k\right)\Id+\underset{l\neq k}{\sum^{n^2}_{k=1}}
a_ka_l\pmb{\alpha}_k\pmb{\alpha}_l
\end{eqnarray}
The anticommuting pairs in the second term of the last equation will vanish, and the above equation may be expressed as:
\begin{equation}
\left(1-\sum^{n^2}_{k=1}a^2_k\right)\Id+2\sum^{}_{k,l>k} a_ka_l\pmb{\alpha}_k\pmb{\alpha}_l=\zerov 
\end{equation}
where the second summation is held only over commuting pairs of basis elements. From the properties of the basis elements, we know that the product of any pair of distinct basis elements is an element of the basis (with exception of the identity matrix that cannot be expressed as a product of any distinct pair of basis elements according to \textbf{Proposition 4}). Thus, the only solution for the above equation is:
\begin{equation}
\sum^{}_{k,l>k} a_ka_l\pmb{\alpha}_k\pmb{\alpha}_l=\zerov;\ \sum^{n^2}_{i=1} a_i^2=1 
\end{equation}
which proves \textbf{Proposition 5}.\end{proof}
An illustrative example is the following:
\begin{equation}
\pmb{\Lambda}=\frac{1}{2}\left(\Rm_1-\Rm_3+\Rm_1\Rm_2+\Rm_2\Rm_3\right). 
\end{equation}
The above example of $\pmb{\Lambda}$ satisfies \textbf{Proposition 5} as the only commuting pairs are $\left\lbrace \Rm_1,\Rm_2\Rm_3\right\rbrace $ and  $\left\lbrace -\Rm_3,\Rm_1\Rm_2\right\rbrace $ and the product of the first pair is the additive inverse of the product of the second pair with $\sum^{}_{i}a^2_i=1$.
\begin{myprop}
The UW-$g$-group decodable codes with single thread weight matrices where the non-zero elements $\in\lbrace \pm 1,\pm j\rbrace$ for $n=2^a$ antennas can exist only for $\Gamma$ sets where the $\pmb{\Lambda}$ matrices are expressed as:
\begin{equation*}
\pmb{\Lambda}=\sum^{n^2}_{k=1}a_k\pmb{\alpha}_k, a_k\in\left\lbrace \frac{n-2\kappa}{n}:\kappa\in\mathbb{N}\right\rbrace,\ \sum^{n^2}_{k=1}a^2_k=1  
\end{equation*}
\end{myprop}
\begin{proof}
see Appendix.
\end{proof}
By using \textbf{Propositions 5} and \textbf{6}, we now have the possibility to exhaustively construct all the possible $\Gamma$ sets that satisfy \textbf{Proposition 3}.


\section{Results}
In this section we provide examples of the application of the proposed method to find the maximum achievable rate of 4$\times$4 UW-$g$-group decodable STBCs where the weight matrices are required to be single thread matrices with non-zero entries $\in \lbrace\pm 1,\pm j\rbrace$. The weight matrices were found through exhaustive computer search.
For the case of four transmit antennas $\left(a=2\right)$, \textbf{Proposition 6} reduces to:
\begin{equation}
\pmb{\Lambda}=\left\lbrace \begin{array}{lc} \pm \pmb{\alpha}_k,& k\in\left\lbrace 1,2,\ldots,16\right\rbrace \\ 
\sum^{4}_{i=1}a_{k_i}\pmb{\alpha}_{k_i},& a_{k_i}\in\left\lbrace \pm \frac{1}{2}\right\rbrace 
\end{array}\right.
\end{equation}
For the symmetric UW-2-group decodable STBCs, we found that the maximum achievable rate is limited to 5/4 cspcu (Complex Symbol Per Channel Use) (see Table~\ref{2g}). However, if the symmetry restriction is relaxed, one can easily obtain rate-$\frac{n^2+1}{2n}$ UW-2-group decodable STBC for $n$ transmit antennas \cite{R>1} giving rise to a rate-17/8 cspcu UW-2-group decodable STBC for four transmit antennas.

For symmetric UW-3-group decodable codes, the maximal achievable rate is proved to be 3/4 cspcu \cite{3SYMS_BND}.
If the symmetry restriction is relaxed, we found that the maximum achievable rate is limited to 1 cspcu (see Table~\ref{3g}). For symmetric UW-4-group decodable codes, it is known that the maximum achievable rate is 1 cspcu \cite{2SYMS_BND}. Examples of these codes may be found in \cite{MULT_GR,MULTI_SYMs,SAST,MDC}.

For the 2-group decodable rate-5/4 code, $g=2,n_1=n_2=5$ thus the decoding complexity for square QAM constellations is of order $\sum^{2}_{i=1}\sqrt{M}^{n_i-1}=2M^2$. For the case of non-rectangular constellations, the rate-5/4 code is no longer 2-group decodable due to the entanglement of the real and imaginary parts of the transmitted complex symbols. In that case, we may use the conditional detection \cite{2TX_MPLX_ORTH} to evaluate the ML estimate of $\left(x_1,\ldots,x_4\right)$ and $\left(x_5,\ldots,x_8\right)$ separately (thanks to the Quasi-orthogonality structure) for a given value of $\left(x_9,x_{10}\right)$. Therefore the decoding complexity is of order $2M^3$. 

For the 3-group decodable rate-1 code, one has $g=3,n_1=n_2=2,n_3=4$, and therefore the worst-case decoding complexity order of square QAM constellations is $\sum^{3}_{i=1}\sqrt{M}^{n_i-1}=2\sqrt{M}+M^{1.5}$. For the case of non-rectangular constellations, the rate-1 3-group decodable code maintains its multi-group decodability structure, but with an increase of decoding complexity order to $\sum^{3}_{i=1}M^{n_i/2}=2M+M^2$. These results are summarized in Table \ref{complexity}.

It is worth noting that the coding gain of the proposed codes is equal to zero, but the full diversity may still be ensured by applying a constellation rotation to each group of symbols, which does not affect the multi-group decodability structure, and hence the decoding complexity remains unchanged. Due to the prohibitive complexity of the numerical optimization of the rotation matrix and specially for high-order constellations, the coding gain optimization and the BER performance may be the subject of subsequent work. 

In addition, the limits on the achievable rates of the special type of UW-$g$-group decodable STBCs which have been investigated in this paper suggest to resort to the so called fast decodable codes which are conditionally $g$-group decodable (e.g. \cite{GLOBECOM11}) thus enabling the use of the conditional detection technique in order to improve the rate/performance/complexity tradeoff.
\section{Conclusion}
In this paper, we addressed the problem of finding the maximum achievable rates of a special type of UW-$g$-group decodable STBCs for $2^a$ transmit antennas. For this purpose, we extended the previously proposed approach of finding UW-2-group decodable codes to search for UW-$g$-group decodable codes in a recursive fashion. The new construction method  was then applied to the type of weight matrices usually proposed in the literature. It was found that the maximum achievable rate for the 4$\times$4 symmetric UW-2-group decodable codes is 5/4 cspcu and that the maximum achievable rate for the 4$\times$4 non-symmetric UW-3-group decodable codes is 1 cspcu. 

\section*{Appendix}
In the following, we will prove \textbf{Proposition 6}.
We have:
\begin{equation}
\text{tr}\left(\Am\right)=0,\ \forall\Am\in \Bc_{2a}\setminus\lbrace j\Id \rbrace.  
\end{equation}
This can be easily verified from \eqref{sigmas} by noting that:
\begin{equation}
\begin{split}
\sigma_1\sigma_2=j\sigma_3;\ \sigma_1\sigma_3&=-j\sigma_2;\ \sigma_2\sigma_3=j\sigma_1\\
\text{and}\  \text{tr}\left(\sigma_i\right)&=0;\ \forall\ i\in\left\lbrace1,2,3\right\rbrace. 
\end{split}
\end{equation}
And from \eqref{reps} we have $\Am\in \Bc_{2a}\setminus\lbrace j\Id \rbrace$ may be expressed as:
\begin{equation}
\Am=\lambda\Xi_1\otimes\Xi_2\ldots\otimes \Xi_a
\end{equation}
where $\lambda \in \left\lbrace \pm 1,\pm j \right\rbrace $, $\Xi_i\in\left\lbrace \sigma_1,\sigma_2,\sigma_3,\Id \right\rbrace $ but $\left\lbrace\Xi_i:i=1,\ldots,a\right\rbrace $ cannot be equal simultaneously to $\Id$ as the set $\Bc_{2^a}$ is linearly independent over $\mathbb{C}$. 
Consequently, one has:
\begin{equation}
\text{tr}\left(\Am\right)=\lambda\text{tr}\left(\Xi_1\otimes\Xi_2\ldots\otimes \Xi_a\right)= \lambda\prod^{a}_{i=1}\text{tr}\left(\Xi_i\right)= 0
\end{equation}
as at least we have $\Xi_k\in\left\lbrace\sigma_1,\sigma_2,\sigma_3\right\rbrace,k\in\left\lbrace 1,\ldots,a\right\rbrace$. Therefore, thanks to \textbf{Proposition 4} we may write:
\begin{equation}
\text{tr}\left(\pmb{\alpha}^H_m\pmb{\alpha}_n\right)\Big\vert_{m\neq n}=\text{tr}\left(\pmb{\alpha}_k \right)\Big\vert_{k\neq 1}=0 
\label{treq}
\end{equation}
Eq~\eqref{treq} may be used to find the coefficients $a_{i=1,2,\ldots,2^{2a}}$ in Eq~\ref{rep1} as below:
\begin{equation}
\text{tr}\left(\pmb{\alpha}^H_k \pmb{\Lambda}\right)=na_k+
\underbrace{\underset{i \neq k}{\sum^{n^2}_{i=1}}a_i\pmb{\alpha}^H_k\pmb{\alpha}_i}_{0}=na_k.
\label{tr1}
\end{equation}
where $n=2^a$. On the other hand, it can be verified that any element of the basis over $\mathbb{C}$ in \eqref{AHbasis} may be expressed as:
\begin{equation}
\pmb{\alpha}_k=\Tm_i\Dm_k,\ i\in\left\lbrace1,2,\ldots n\right\rbrace,k\in\left\lbrace1,2,\ldots,n^2 \right\rbrace  
\end{equation} 
\begin{table*}
\setcounter{equation}{48}
\begin{equation}
\pmb{\Lambda}=\left\lbrace \begin{array}{ll} \Tm_{k_1}\Dm_{k_1} & k_1\in\left\lbrace1,2,\ldots,n\right\rbrace \\
\Tm_{k_1}\Dm_{k_1}+\Tm_{k_2}\Dm_{k_2}&k_1,k_2\in\left\lbrace 1,2,\ldots,n\right\rbrace,\ k_1\neq k_2\\ 
\Tm_{k_1}\Dm_{k_1}+\Tm_{k_2}\Dm_{k_2}+\Tm_{k_3}\Dm_{k_3}&k_1,k_2,k_3\in\left\lbrace1,2,\ldots,n\right\rbrace,\ k_1\neq k_2\neq k_3\\ 
\vdots&\vdots\\
\sum^{n/2}_{i=1}\Tm_{k_i}\Dm_{k_i}&k_1,k_2,\ldots k_{n/2}\in\left\lbrace1,2,\ldots,n\right\rbrace,\ k_1\neq k_2\neq \ldots k_{n/2}
 \end{array}\right.
\label{Lambda}
\end{equation}
\hrule
\end{table*}
where $\Tm_i$ is one of the disjoint permutation matrices that indicate the threads occupied by the basis elements \eqref{AHbasis} and $\Dm_k$ is a diagonal matrix with entries $\in\left\lbrace\pm 1,\pm j\right\rbrace$. For instance, the four permutation matrices for the case of four transmit antennas denoted  $\Tm_1,\Tm_2,\Tm_3,\Tm_4$ are shown in Table~\ref{perms}. It follows directly from the properties of basis elements \eqref{AHbasis} that the matrices $\Tm_i,\ i\in\left\lbrace1,2,\ldots n\right\rbrace$ have the following properties:
\setcounter{equation}{46}
\begin{eqnarray}
\Tm^T_k&=&\Tm_k\\
\Tm_{k_1}\Tm_{k_2}&=&\left\lbrace \begin{array}{ll}\Tm_{k_3}&k_1\neq k_2\neq k_3\in\left\lbrace 2,3,\ldots n\right\rbrace\\
\Id & k_1=k_2  \end{array}\right.
\end{eqnarray}
From the definition of the matrix $\pmb{\Lambda}$ in Eq ~\ref{lambdakl_1}, it can be verified that restricting the weight matrices $\Am$ and $\Bm$ to single-thread matrices with entries $\in \left\lbrace\pm 1,\pm j\right\rbrace$ turns out that the corresponding matrix $\pmb{\Lambda}$ is evenly a single-thread matrix with entries $\in \left\lbrace\pm 1,\pm j\right\rbrace$, and thus may expressed according to \eqref{Lambda}. Where $\Dm_{k_i},\ k_i\in\left\lbrace1,2,\ldots,n\right\rbrace$ are diagonal matrices with entries $\in \left\lbrace 0,\pm 1,\pm j\right\rbrace$ such that $\Dm_{k_m}\Dm_{k_n}=\zerov,\ \forall\ k_m\neq k_n$. This is because the matrix $\pmb{\Lambda}$ is required to be anti-hermitian (see \textbf{Proposition 2}), then it cannot have an odd number of common positions with threads $\Tm_1,\Tm_2,\ldots,\Tm_n$. By common positions with $\Tm_k$, we mean the number entries in the matrix $\pmb{\Lambda}$ that corresponds to a non-zero entry in $\Tm_k$. This can be verified by noting that the matrices $\Tm_1,\Tm_2,\ldots,\Tm_n$ are symmetric and disjoint. Moreover, the matrix $\pmb{\Lambda}$ is required to be a single-thread matrix. For the case of four transmit antennas, the above reduces to:
\setcounter{equation}{49}
\begin{equation*}
\pmb{\Lambda}=\left\lbrace \begin{array}{ll} \Tm_{k_1}\Dm_{k_1} & k_1\in\left\lbrace1,2,3,4\right\rbrace \\
\Tm_{k_1}\Dm_{k_1}+\Tm_{k_2}\Dm_{k_2}&k_1,k_2\in\left\lbrace 1,2,3,4\right\rbrace,\ k_1\neq k_2\\ 
\end{array}\right.
\end{equation*}
It is worth noting that the matrix $\pmb{\Lambda}$ being single-thread implies that in the first case all of the diagonal elements of the matrix $\Dm_{k1}$ are strictly non-zero and in the second case each of the diagonal matrices $\Dm_{k_1}$ and $\Dm_{k_2}$ has strictly two non-zero elements such that $\Dm_{k_1}\Dm_{k_2}=\zerov$. Let $\pmb{\alpha}_k=\Tm_m\Dm_m$, thus:
\begin{equation}
\begin{split}
\text{tr}\left(\pmb{\alpha}^H_k\pmb{\Lambda}\right)&=
\sum^{n/2}_{i=1}\text{tr}\left(\Tm_m\Tm_{k_i}\Dm_{k_i}\Dm^H_m\right)\\
&=\left\lbrace \begin{array}{ll}0&\ m\neq k_1\neq k_2\neq\ldots k_{n/2}\\
\text{tr}\left(\Dm_{k_i}\Dm^H_m\right)& k_i=m
\end{array}\right.
\end{split}
\label{tr2}
\end{equation}
Recalling that $a_k$ is restricted to be real (see \textbf{Proposition 5}) and equating \eqref{tr1} and \eqref{tr2} one easily obtains:
\begin{equation}
a_k\in\left\lbrace \frac{n-2\kappa}{n}:\ \kappa\in\mathbb{N}\right\rbrace
\end{equation}
On the other hand, according to \textbf{Proposition 5} one has $\sum^{n^2}_{k=1}a^2_k=1$ thus completing the proof.

\begin{table}[h!]
\topcaption{The basis elements of the 4$\times$4 matrices over $\mathbb{C}$}
\centering
\begin{tabular}{|c|c|c|c|}
\hline
$\pmb{\alpha}_1=j\Id$&$\pmb{\alpha}_2=\Rm_1,\nonumber$&$\pmb{\alpha}_3=\Rm_2$&$\pmb{\alpha}_4=\Rm_3$\\
\hline
$\pmb{\alpha}_5=\Rm_4$&$\pmb{\alpha}_6=\Rm_1\Rm_2$&$\pmb{\alpha}_7=\Rm_1\Rm_3$&$\pmb{\alpha}_8=\Rm_1\Rm_4$\\
\hline
$\pmb{\alpha}_9=\Rm_2\Rm_3$&$\pmb{\alpha}_{10}=\Rm_2\Rm_4$&$\pmb{\alpha}_{11}=\Rm_3\Rm_4$&$\pmb{\alpha}_{12}=j\Rm_1\Rm_2\Rm_3$\\
\hline
$\pmb{\alpha}_{13}=j\Rm_1\Rm_2\Rm_4$&$\pmb{\alpha}_{14}=j\Rm_1\Rm_3\Rm_4$&$\pmb{\alpha}_{15}=j\Rm_2\Rm_3\Rm_4$&$\pmb{\alpha}_{16}=j\Rm_1\Rm_2\Rm_3\Rm_4$\\
\hline
\end{tabular}
\label{basis_mat}
\end{table}

\begin{table}[h!]
\topcaption{The weight matrices of the rate-5/4 UW-2-group decodable code}
\centering
\begin{tabular}{|c|c|}
\hline
$\Am_1= \begin{bmatrix}0&0&0&1\\0&1&0&0\\1&0&0&0\\0&0&1&0\end{bmatrix}$&
$\Bm_1= \begin{bmatrix}0&0&0&-j\\0&j&0&0\\j&0&0&0\\0&0&-j&0\end{bmatrix}$\\
$\Am_2= \begin{bmatrix}0&0&0&1\\0&1&0&0\\-1&0&0&0\\0&0&-1&0\end{bmatrix}$&
$\Bm_2= \begin{bmatrix}0&0&0&-j\\0&j&0&0\\-j&0&0&0\\0&0&j&0\end{bmatrix}$\\
$\Am_3= \begin{bmatrix}0&0&0&1\\0&1&0&0\\-1&0&0&0\\0&0&1&0\end{bmatrix}$&
$\Bm_3= \begin{bmatrix}0&0&0&-j\\0&-j&0&0\\j&0&0&0\\0&0&-j&0\end{bmatrix}$\\
$\Am_4= \begin{bmatrix}0&0&0&1\\0&1&0&0\\0&0&-j&0\\-j&0&0&0\end{bmatrix}$&
$\Bm_4= \begin{bmatrix}0&-1&0&0\\0&0&0&1\\j&0&0&0\\0&0&-j&0\end{bmatrix}$\\
$\Am_5= \begin{bmatrix}0&0&0&1\\0&1&0&0\\0&0&1&0\\-1&0&0&0\end{bmatrix}$&
$\Bm_5= \begin{bmatrix}0&-j&0&0\\0&0&0&-j\\j&0&0&0\\0&0&-j&0\end{bmatrix}$\\
\hline
\end{tabular}
\label{2g}
\end{table}

\begin{table}[h!]
\topcaption{The weight matrices of the rate-1 UW-3-group decodable code}
\centering
\begin{tabular}{|c|c|cc|}
\hline
$\Am_1=\begin{bmatrix}-j&0&0&0\\0&j&0&0\\0&0&-1&0\\0&0&0&1\end{bmatrix}$&
$\Bm_1=\begin{bmatrix}0&-1&0&0\\1&0&0&0\\0&0&0&j\\0&0&-j&0\end{bmatrix}$&
$\Cm_1=\begin{bmatrix}1&0&0&0\\0&1&0&0\\0&0&0&1\\0&0&1&0\end{bmatrix}$;&
$\Cm_2=\begin{bmatrix}0&j&0&0\\j&0&0&0\\0&0&0&1\\0&0&1&0\end{bmatrix}$\\
$\Am_2=\begin{bmatrix}-j&0&0&0\\0&j&0&0\\0&0&1&0\\0&0&0&-1\end{bmatrix}$&
$\Bm_2=\begin{bmatrix}0&-1&0&0\\1&0&0&0\\0&0&0&-j\\0&0&j&0\end{bmatrix}$&
$\Cm_3=\begin{bmatrix}0&-j&0&0\\-j&0&0&0\\0&0&0&1\\0&0&1&0\end{bmatrix}$;&
$\Cm_4=\begin{bmatrix}1&0&0&0\\0&1&0&0\\0&0&-j&0\\0&0&0&-j\end{bmatrix}$\\
\hline
\end{tabular}
\label{3g}
\end{table}

\begin{table}[h!]
\centering
\topcaption{Summary of results}
\begin{tabular}{|c|c|c|c|c}
\hline
number of groups&maximum rate&\multicolumn{2}{c|}{Complexity order of}\\
\cline{3-4} &&square QAM&non-rectangular QAM\\
\hline
2 (sym)&5/4&$2M^2$&$2M^3$\\
\hline
2 (non-sym)&17/8 \cite{R>1}&$M^{5.5}$&$6M^{6.5}$\\
\hline
3 (sym)&3/4 \cite{3SYMS_BND}&$3\sqrt{M}$&$3M$\\
\hline
3 (non-sym)&1&$2\sqrt{M}+M^{1.5}$&$2M+M^2$\\
\hline
\end{tabular}
\label{complexity}
\end{table}

\begin{table}[h!]
\centering
\topcaption{The four permutation matrices for the case of 4$\times$4 matrices}
\begin{tabular}{|c|c|c|c|}
\hline
$\Tm_1=\begin{bmatrix}1&0&0&0\\0&1&0&0\\0&0&1&0\\0&0&0&1\end{bmatrix}$&
$\Tm_2=\begin{bmatrix}0&1&0&0\\1&0&0&0\\0&0&0&1\\0&0&1&0\end{bmatrix}$&
$\Tm_3=\begin{bmatrix}0&0&1&0\\0&0&0&1\\1&0&0&0\\0&1&0&0\end{bmatrix}$&
$\Tm_4=\begin{bmatrix}0&0&0&1\\0&0&1&0\\0&1&0&0\\1&0&0&0\end{bmatrix}$\\
\hline
\end{tabular}
\label{perms}
\end{table}

\end{document}